\documentclass[10pt,A4paper,conference]{IEEEtran}
\usepackage{times}
\usepackage{algorithmic}
\usepackage{amsmath}
\usepackage{amssymb}
\usepackage{amsthm}
\usepackage{color}
\usepackage{colortbl}
\usepackage{float}
\usepackage{cite}
\usepackage{subfig}
\usepackage{verbatim}
\usepackage[pdftex]{graphicx}

\usepackage{slashbox}

\newtheorem{theorem}{Theorem}
\newtheorem{lemma}[theorem]{Lemma}

\newtheorem{corollary}[theorem]{Corollary}

\newtheorem{example}{Example}
\newtheorem{remark}{Remark}

\newcommand{\deff}{\mbox{$\stackrel{\rm def}{=}$}}
\newcommand{\field}[1]{\mathbb{#1}}

\newcommand{\F}{\field{F}}

\newcommand{\cB}{{\mathcal B}}

\newcommand{\cD}{{\mathcal D}}
\newcommand{\cG}{{\mathcal G}}

\newcommand{\cP}{{\mathcal P}}

\newcommand{\sP}{\cP}
\newcommand{\sG}{\cG}

\newcommand{\Gr}{\smash{{\sG\kern-1.5pt}_q\kern-0.5pt(n,k)}}
\newcommand{\Gk}{\smash{{\sG\kern-1.5pt}_q\kern-0.5pt(n,k_1)}}
\newcommand{\Gkk}{\smash{{\sG\kern-1.5pt}_q\kern-0.5pt(n,k_2)}}
\newcommand{\Grtwo}{\smash{{\sG\kern-1.5pt}_2\kern-0.5pt(n,k)}}
\newcommand{\Gkone}{\smash{{\sG\kern-1.5pt}_q\kern-0.5pt(n,k_1)}}
\newcommand{\Gktwo}{\smash{{\sG\kern-1.5pt}_q\kern-0.5pt(n,k_2)}}
\newcommand{\Ps}{\smash{{\sP\kern-2.0pt}_q\kern-0.5pt(n)}}

\begin{document}

\title{Optimal Fractional Repetition Codes\\ and Fractional Repetition Batch Codes}

\author{
%\IEEEauthorblockN{Natalia Silberstein}
%\IEEEauthorblockA{School of Electrical and\\Computer Engineering\\
%Georgia Institute of Technology}
%\and
%\IEEEauthorblockN{Homer Simpson}
%\IEEEauthorblockA{Twentieth Century Fox\\
%Springfield, USA}
%\and
\IEEEauthorblockN{Natalia Silberstein and Tuvi Etzion}
\thanks{This research was supported in part  by the Israeli Science Foundation (ISF), Jerusalem, Israel, under Grant 10/12.}
  \thanks{The first author was supported in part at the Technion by a Fine Fellowship.}
\IEEEauthorblockA{
Computer Science Department\\
Technion-Israel Institute of Technology\\
Haifa 32000, Israel\\
Email: \{natalys, etzion\}@cs.technion.ac.il\\}
}

\maketitle

\begin{abstract}
%\boldmath
Fractional repetition (FR) codes is a family of codes for distributed storage systems (DSS) that allow uncoded exact repairs with minimum repair bandwidth. In this work, we consider a bound on the maximum amount of data that can be stored using an FR code. Optimal FR codes which attain this bound are presented. The constructions of these FR codes are based on families of regular graphs, such as Tur\'an graphs and graphs with large girth; and  on combinatorial designs, such as transversal designs and generalized polygons.
In addition, based on a connection between FR codes and batch codes, we propose a new family of codes for DSS, called fractional repetition batch codes, which allow uncoded efficient exact repairs and load balancing which can be performed by several users in parallel.
\end{abstract}

%*********************************************************************************************************************************************************
%*********************************************************************************************************************************************************
%                                                         Introduction
%*********************************************************************************************************************************************************
%*********************************************************************************************************************************************************

\section{Introduction}
In distributed storage systems, data is stored across a network of nodes, which can unexpectedly fail. To provide reliability, data redundancy based on coding techniques is introduced in such systems. Moreover, existing erasure codes allow to minimize  the storage overhead.
In~\cite{dimakis} Dimakis et al. introduced a new family of erasure codes, called \emph{regenerating codes}, which allow efficient single node repairs.  In particular, they presented two families of regenerating codes, called \emph{minimum storage regenerating} (MSR) codes and \emph{minimum bandwidth regenerating} (MBR) codes, which correspond to the two extreme points on the storage-bandwidth trade-off~\cite{dimakis}.
An $(n,k,d,\alpha,\beta)_q$ regenerating code~$C$, where $k\leq d\leq n-1$, $\beta\leq \alpha$, is used to store a file in $n$ nodes; each node stores $\alpha$ symbols from $\F_q$, the finite field with $q$ elements,
such that the stored file can be recovered by downloading  the data from any set of $k$ nodes.
 %Note, that this means that any $n-k$ node failures (i.e., erasures) can be corrected by this code.
 When a single node fails, a newcomer node which substitutes the failed node contacts  with a random set of $d$ other nodes and downloads $\beta$ symbols of each node in this set to reconstruct the failed data. This process is called a \emph{node repair}, and the amount of data downloaded to repair a failed node, $\beta d$, is called the \emph{repair bandwidth}.

In~\cite{Rashmi09,RSKR12} Rashmi et al. presented a construction for MBR codes which have the additional property of exact \emph{repair by transfer}, or exact \emph{uncoded repair}. In other words, the $(n,k,d=n-1,M=k\alpha-\binom{k}{2},\alpha=n-1,\beta=1)$ code proposed in~\cite{Rashmi09,RSKR12} allows efficient exact node repairs where no decoding is needed. Every node participating in a node repair process just passes one symbol
which will be  directly stored in the newcomer node.
%Each node stores $\alpha=n-1$ symbols and the number of nodes which participate in a node repair process is also $d=n-1$.
This construction is based on a concatenation of an outer MDS code with an inner repetition code based on a complete graph.
El Rouayheb and Ramchandran~\cite{RoRa10} generalized the construction of~\cite{Rashmi09} and defined a new family of codes for DSS which allow exact repairs by transfer for a wide range of parameters. These codes, called \emph{DRESS} (Distributed Replication based Exact Simple Storage) codes~\cite{PNRR11}, consist of the concatenation of an outer MDS code and the inner repetition code called \emph{fractional repetition} (FR) code.
However, in contrast
to MBR codes, where a random set of size $d$ of
available nodes is used for a node repair, the repairs with DRESS codes are table based. This usually allows to store more data compared to MBR codes.

Constructions of FR codes  based on some regular graphs and combinatorial designs
can be found for example in~\cite{KoGi11,RoRa10,OlRa12,OlRa13}. However, the optimality of the constructed FR codes regarding the FR capacity, i.e. the maximality of the size of the stored file, was not considered.

In this work, we address the problem of constructing optimal FR codes and hence, optimal DRESS codes.
%First, we propose  constructions for FR codes  which
%are based on different families of regular graphs, specifically, Tur\'an graphs and cage graphs. Next, we consider FR codes  based on a family of combinatorial designs, called transversal designs. This construction generalizes the construction based on Tur\'an graphs. Another construction is based on biregular bipartite graphs with a given girth. One important family of such graphs are the generalized polygons.
Moreover, based on a connection between FR codes and combinatorial batch codes,
we propose a new family of codes for DSS, called fractional repetition batch (FRB) codes, which enable uncoded repairs and load balancing that can be performed by several users in parallel.
%We present examples of constructions of FRB codes based on some graphs and combinatorial designs.

The rest of the paper is organized as follows.
In Section~\ref{sec:DRESS} we define DRESS codes and FR codes based on regular graphs and combinatorial designs.
In Section~\ref{sec:rho2} we present optimal FR codes  based on Tur\'an graphs and on graphs with large girth.
In Section~\ref{sec:rho>2} we consider optimal FR codes based on transversal designs and on generalized polygons.
%, for the reconstruction degree $\rho=2$ and $\rho>2$, respectively.
In Section~\ref{sec:FRB} we define FRB codes and present some examples for their constructions.
Conclusion is given in Section~\ref{sec:conclusion}.
We point out that,
throughout this paper, proofs are often omitted due to space
limitations. Details of all the proofs can be found in~\cite{SiEtFR}.

%*********************************************************************************************************************************************************
%*********************************************************************************************************************************************************
%                                                         optimal FR codes
%*********************************************************************************************************************************************************
%*********************************************************************************************************************************************************

%\section{Optimal FR Codes}
%\label{sec:optimalFR}
%*********************************************************************************************************************************************************
%*********************************************************************************************************************************************************
%                                                         DRESS and FR codes
%*********************************************************************************************************************************************************
%*********************************************************************************************************************************************************

\section{Preliminaries}
\label{sec:DRESS}
 An $(n,\alpha, \rho)$ \emph{FR code} $C$ %with repetition degree $\rho$
is a collection of $n$ subsets $N_1,\ldots, N_n$ of $[\theta]\deff\{1,2,\ldots,\theta\}$, $n\alpha=\rho \theta$,  such that
\begin{itemize}
  \item $|N_i|=\alpha$ for  each $i$, $1\leq i\leq n$;
  \item each symbol of $[\theta]$ belongs to exactly $\rho$ subsets in $C$, where $\rho$ is called the \emph{repetition degree} of $C$.
\end{itemize}
A $\left[(\theta, M), k, (n,\alpha, \rho)\right]$ \emph{DRESS code} is a code obtained by the concatenation of an outer $(\theta, M)$ MDS code and an inner $(n,\alpha, \rho)$ FR code $C$. To store a file $\textbf{f}\in \F_q^M$  in a DSS, $\textbf{f}$  is first encoded by using the MDS code; next, the $\theta$ symbols of the codeword $\textbf{c}_{\textbf{f}}$ from the MDS code, which encodes the file $\textbf{f}$, are placed in the $n$ nodes defined by $C$, as follows: node $i\in [n]$ of the DSS stores $\alpha$ symbols of $\textbf{c}_\textbf{f}$, indexed by the elements of the subset $N_i$.
The encoding scheme for a DRESS code is shown in Fig.~\ref{fig:FRscheme}.
\begin{figure}[t]
\centering
\includegraphics[trim=0 0 0 70,clip=true,width=0.99\columnwidth]{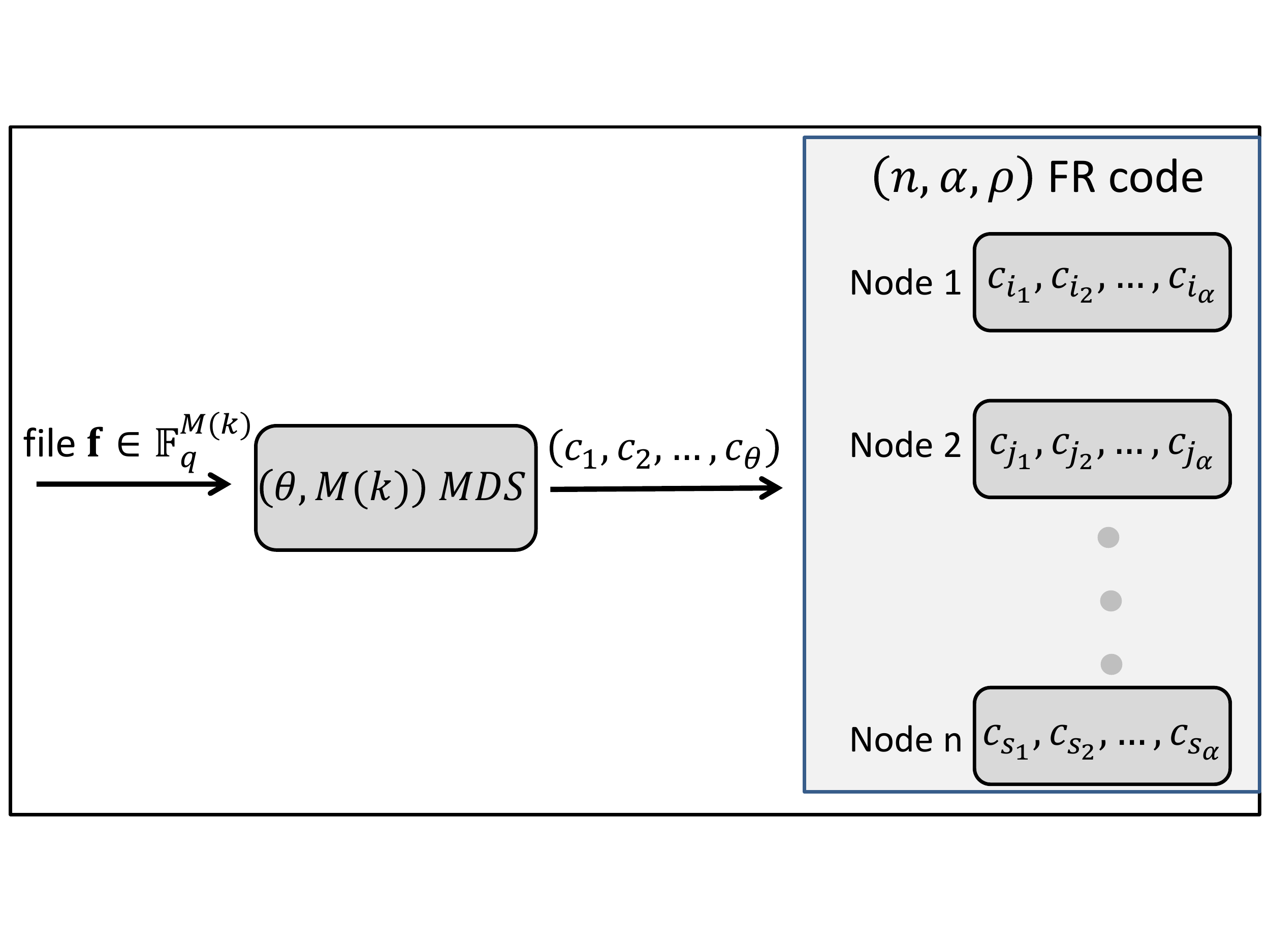}
\caption{The encoding scheme for a DRESS code}\label{fig:FRscheme}
\end{figure}

Each symbol of $\textbf{c}_\textbf{f}$ is stored in exactly $\rho$ nodes. It should be possible to reconstruct the stored file $\textbf{f}$ of size $M$ from any set of $k$ nodes, and hence,
\begin{equation}
\label{eq:rate}
M\leq \min_{|I|=k}|\cup_{i\in I}N_i|.
\end{equation}
Since we want to maximize the size of a file that can be stored by using  a DRESS code, in the sequel we will always assume that $M=\min_{|I|=k}|\cup_{i\in I}N_i|$.
Note, that the same FR code can be used in different DRESS codes, with different $k$'s as reconstruction degrees, and different MDS codes. The file size $M$, which is the dimension of the chosen MDS code,  depends on the value of $k$ and hence in the sequel we will use $M(k)$ to denote the size of the file.
An $(n,\alpha, \rho)$  FR code is called \emph{universally good}~\cite{RoRa10} if for any $k\leq \alpha$ the $\left[(\theta, M(k)), k, (n,\alpha, \rho)\right]$ DRESS code satisfies
\begin{equation}
\label{eq:univGood}
M(k)\geq k\alpha-\binom{k}{2},
\end{equation}
where the  righthand side of equation~(\ref{eq:univGood}) is the maximum file size (called \emph{MBR capacity}) that can be stored using an MBR code~\cite{dimakis}.
 Note also that if an FR code $C$ is universally good
 %satisfies~(\ref{eq:univGood})
  then $|N_i\cap N_j|\leq 1$, for $N_i,N_j\in C$, $i\neq j\in[n]$~\cite{Rashmi09}.
In the sequel, we will consider only universally good FR codes.

%When some node $j$ fails, it can be repaired by using  a set of $d=\alpha$ other nodes  $i_1,i_2\ldots,i_{\alpha}$, such that
%$N_j\cap N_{i_s}\neq \varnothing$, $s\in [\alpha]$, and $\cup_{s\in [\alpha]}(N_j\cap N_{i_s})=N_j$ . Each such node passes exactly one symbol ($\beta=1$) to repair node $j$. Note that the repair bandwidth of a DRESS code is the same as the repair bandwidth of an MBR code.
%%The encoding scheme based on an FR code is shown in Fig.~\ref{fig:FRscheme}.

An upper bound on the maximum file size $M(k)$ of a $\left[(\theta, M(k)), k, (n,\alpha, \rho)\right]$ DRESS code ($n\alpha=\rho\theta$), called the \emph{FR capacity} and denoted in the sequel by $A(n,k,\alpha, \rho)$, was presented in~\cite{RoRa10}:
%\footnote{In~\cite{RoRa10} this value is called the \emph{FR capacity}.}:
  \begin{equation}\label{eq:bound2}
   A(n,k,\alpha, \rho)\leq  \varphi(k), \textmd{ where } \varphi(1)=\alpha,
  \end{equation}
  \[\textmd{  } \varphi(k+1)=\varphi(k)+\alpha-\left\lceil\frac{\rho \varphi(k)-k\alpha}{n-k}\right\rceil.\]
   Note  that for any given $k$, the function $A(n,k,\alpha, \rho)$ is determined by the parameters of the inner FR code.
   We call an FR code \emph{$k$-optimal} if a file stored by using this code is the maximum possible for the given $k$.
   We call an FR code \emph{optimal} if  for any $k\leq \alpha$  it is $k$-optimal.

Let $C$ be an $(n,\alpha,\rho)$ FR code. $C$ can be described by an \emph{incidence matrix} $\textbf{I}(C)$, which is an $n\times \theta$ binary matrix, $\theta=\frac{n\alpha}{\rho}$, whose rows indexed by the nodes and columns indexed by the symbols of the corresponding MDS codeword, such that $(\textbf{I}(C))_{i,j}=1$ if and only if node $i$ contains symbol~$j$.
Note that  every row of $\textbf{I}(C)$ has $\alpha$ \emph{ones} and every column of $\textbf{I}(C)$ has $\rho$ \emph{ones}.

Let $G=(V,E)$ be an $\alpha$-regular graph with $n=|V|$ vertices. We say that an $(n,\alpha,\rho=2)$ FR code $C$ \emph{is based on $G$} if $\textbf{I}(C)=\textbf{I}(G)$, where $\textbf{I}(G)$ is the $|V|\times |E|$ incidence matrix of $G$. Such a code will be denoted by $C_G$.
%It can be readily verified that any $(n,\alpha,2)$ FR code can be represented by an $\alpha$-regular graph with $n$ vertices.

Let $\cD=(\cP,\cB)$ be a design with $|\cP|=n$ points such that each block $B\in\cB$ contains $\rho$ points and each point $p\in \cP$ is contained in $\alpha$ blocks. We say that an $(n,\alpha,\rho)$ FR code $C$ \emph{is based on $\cD$} if $\textbf{I}(C)=\textbf{I}(\cD)$, where $\textbf{I}(\cD)$ is the $|\cP|\times |\cB|$ incidence matrix of $\cD$. Such a code will be denoted by $C_{\cD}$.

%*********************************************************************************************************************************************************
%*********************************************************************************************************************************************************
%                                                         rho=2
%*********************************************************************************************************************************************************
%*********************************************************************************************************************************************************

\section{Optimal FR Codes with Repetition Degree $\rho=2$}
\label{sec:rho2}
In this section we consider optimal FR codes with repetition degree 2.
First, we present the following useful lemma which shows a connection between the problem of finding the maximum file size of an FR code based on a graph and the edge isoperimetric problem on graphs~\cite{Bez99}.

\begin{lemma}
\label{lm:isoperimetric}
Let $G=(V,E)$ be an $\alpha$-regular graph and let $C_G$ be the FR code based on $G$. We denote by $G_k$ the family of induced subgraphs of $G$ with $k$ vertices.
%$$G_k=\{G'=(V',E'):|V'|=k, G'\textmd{ is an induced subgraph of }G\}.$$
Then the file size $M(k)$ of $C_G$ is given by
$$M(k)=k\alpha -\max_{G'=(V',E')\in G_k}|E'|.
$$
\end{lemma}

\begin{proof}
For each induced subgraph $G'=(V',E')\in G_k$ we define  $E'_{\textmd{cut}}$ to be the set of all the edges of $E$ in the cut between $V'$ and $V\setminus V'$, i.e.,
$$E'_{\textmd{cut}}=\{\{v,u\}\in E: v\in V', u\in V\setminus V'\}.
$$
Clearly,  $k\alpha=2|E'|+|E'_{\textmd{cut}}|$ for every $G'\in G_k$.
Note that $M(k)=\min_{G'\in G_k}\{|E'|+|E'_{\textmd{cut}}|\}$ and hence
$$M(k)=\min_{G'\in G_k}\{|E'|+\alpha k-2|E'|\}=\alpha k-\max_{G'\in G_k}\{|E'|\}.
$$
\end{proof}
The following lemma directly follows from Lemma~\ref{lm:isoperimetric}.
\begin{lemma}
\label{lm:clique} Let $G$ be an $\alpha$-regular graph with $n$ vertices, and let $M(k)$ be the file size of the corresponding code~$C_{G}$.
The graph $G$ contains a $k$-clique if and only if $M(k)=k\alpha-\binom{k}{2}$.
\end{lemma}
\begin{corollary}
 \label{cor:clique}
 The file size $M(k)$ of an FR code $C_G$, where $G$  is a graph which does not contain a $k$-clique, is strictly larger than the MBR capacity.
\end{corollary}

One of the main advantages of an FR code is that its file size usually exceeds the MBR capacity. Hence, as a consequence of Corollary~\ref{cor:clique}, we consider regular graphs which do not contain a $k$-clique for a given $k$. In particular,
%It is well known, that Tur\'an graphs do not contain a  clique of a given size and also have the minimum number of vertices among the graphs with this property.
we consider a family of regular graphs, called \emph{Tur\'an graphs},  which do not contain a  clique of a given size and also have the smallest number of vertices~\cite{Jukna}.
Let $r,n$ be two integers such that $r$ divides~$n$. An $(n,r)$-\emph{Tur\'an graph} is defined as a regular complete $r$-partite graph, i.e., a graph formed by partitioning the set of~$n$ vertices into $r$ parts of size $\frac{n}{r}$ and connecting each two vertices of different parts by an edge.  Clearly, an $(n,r)$-Tur\'an graph does not contain a clique of size $r+1$  and it is an $(r-1)\frac{n}{r}$-regular graph.

%Therefore, since Tur\'an graphs have the minimum number of vertices among the graphs which
%do not contain a  clique of a given size (see Corollary~\ref{cor:Turan}),  we consider FR codes based on Tur\'an graphs.
The following theorem shows that FR codes obtained from Tur\'an graphs attain the upper bound in~(\ref{eq:bound2}) for all $k\leq \alpha$ and hence they are optimal FR codes. The proof of this theorem follows from Lemma~\ref{lm:isoperimetric} and by Tur\'an's theorem~\cite[p. 58]{Jukna}.
\begin{theorem}
\label{thm:alternativeTuran}
Let $T=(V,E)$ be an $(n,r)$-Tur\'an graph, $r<n$, $\alpha=(r-1)\frac{n}{r}$, and let $k$ be an integer such that $1\leq k\leq \alpha$. Then the $(n,\alpha,2)$ FR code $C_{T}$ based on $T$ has file size given by
\begin{equation}
\label{eq:altTuran}
M(k)=k\alpha-\left\lfloor\frac{r-1}{r}\cdot\frac{k^2}{2}\right\rfloor
\end{equation}
 which attains the upper bound in~(\ref{eq:bound2}).
\end{theorem}
\begin{comment}
\begin{proof}
From Lemma~\ref{lm:isoperimetric} and by Tur\'an's theorem~\cite{Godsil} it directly follows that
%We consider an induced  subgraph $G'=(V',E')$ of the Tur\'an graph $T$ with $|V'|=k$ vertices which have the maximum number of edges. Since $G'$ is a subgraph of $T$, in particular it does not contain a clique of size $r+1$. Then, by Tur\'an's theorem~\cite{Godsil}, $|E'|\leq \frac{r-1}{r}\frac{k^2}{2}$. Hence by Lemma~\ref{lm:isoperimetric}, assuming that $M(k)$ is an integer, we have
\[M(k)=k\alpha-\left\lfloor\frac{r-1}{r}\cdot\frac{k^2}{2}\right\rfloor.
\]
 In addition, one can verify (by induction) that for the parameters of the constructed code $C_T$ the bound in~(\ref{eq:bound2}) equals to~(\ref{eq:altTuran}).
\end{proof}
\end{comment}
Note that an $(n-1)$-regular complete graph $K_n$ is an $(n,n)$-Tur\'an graph. Hence, the construction of MBR codes from~\cite{Rashmi09,RSKR12} can be considered as a special case of our construction of the DRESS codes with an inner FR code based on a Tur\'an graph.
Note also that an $\alpha$-regular complete bipartite graph $K_{\alpha,\alpha}$ is a $(2\alpha,2)$-Tur\'an graph. The following example illustrates Theorem~\ref{thm:alternativeTuran} for such a graph.

\begin{example}
\label{ex:bipartite}
 The $(6,3,2)$ FR code based on $K_{3,3}$ and its file size for $1\leq k\leq 3$ are shown in Fig.~\ref{fig:bipartite}.
\begin{figure}[t]
 \centering
 \includegraphics[trim=0 0 0 40,clip=true,width=0.99\columnwidth]{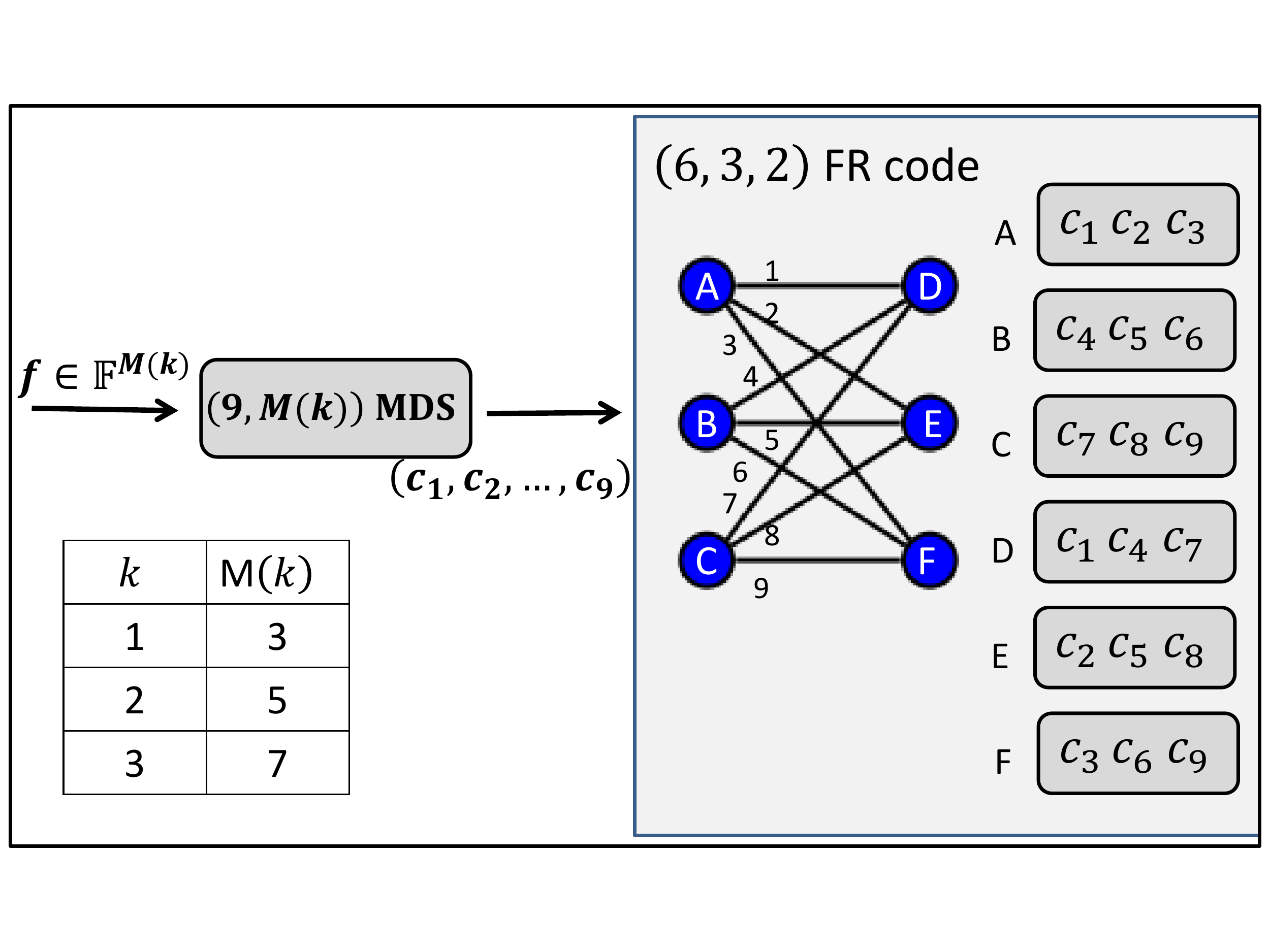}
 \caption{The $((9,M(k)),k,(6,3,2))$ DRESS code with the inner  FR code based on the complete bipartite graph $K_{3,3}$}\label{fig:bipartite}
\end{figure}
\end{example}

%\begin{remark}
%One can prove that for any  $k> r$, the file size of the code $C_{T}$  is strictly larger than the MBR capacity, i.e.,
%$$M(k) > k\alpha-\binom{k}{2}.
%$$
%\end{remark}
The proof of the following lemma can be easily verified from Lemma~\ref{lm:isoperimetric}.
\begin{lemma}
\label{lm:upper} Let $C$ be an $(n,\alpha, 2)$ FR code. Then the file size $M(k)$ of $C$ for any $1\leq k\leq \alpha$ satisfies
$$M(k)\leq k\alpha-k+1.
$$
\end{lemma}

By Lemma~\ref{lm:isoperimetric},  to obtain a large value for $M(k)$, every induced subgraph with $k$ vertices should be as sparse as possible. Hence, for the rest of this section
we consider graphs where the induced subgraphs with $k$ vertices, $1\leq k\leq \alpha$, will be cycle-free. These are graphs with girth at least $k+1$, where the \emph{girth} of a graph is the length of its shortest cycle.

\begin{lemma}
\label{lm:girth}
 Let $G$ be an $\alpha$-regular graph with $n$ vertices and let $M(k)$ be  the file size of the corresponding FR code $C_{G}$.
The girth of  $G$ is at least $k+1$ if and only if $M(k)=k\alpha -(k-1)$.
\end{lemma}
\begin{comment}
\begin{proof}
Let $G$ be a graph with girth $g$.
 Any induced subgraph $G'$ of $G$ with $k$ vertices has at most $k-1$ edges if and only if $g\geq k+1$. Clearly, there exists at least one induced subgraph $G'$ of $G$ with $k$ vertices and $k-1$ edges. Thus, by Lemma~\ref{lm:isoperimetric},
$M(k)=k\alpha-(k-1)
$.
\end{proof}
\end{comment}
\begin{corollary}
\label{cor:girthOptimal}
For each $k\leq g-1$,
an FR code $C_G$ based on an $\alpha$-regular graph~$G$ with girth $g$ attains the bound in~(\ref{eq:bound2}), and hence it is $k$-optimal. $C_G$ also attains the bound of Lemma~\ref{lm:upper}.
\end{corollary}
\begin{corollary}
\label{cor:Optimal}
An FR code $C_G$ based on an $\alpha$-regular graph~$G$ with girth $g\geq \alpha+1$ is optimal.
 \end{corollary}

\begin{comment}
\begin{proof}Directly follows from  $\left\lceil\frac{\rho g(k)-k\alpha}{n-k}\right\rceil\geq 1$ in~(\ref{eq:bound2}).
\end{proof}
\end{comment}

The proof of the following theorem follows from Lemma~\ref{lm:girth} and the fact that any two cycles in a graph with girth $g$ have at most $\lfloor g/2\rfloor+1$ common vertices.
\begin{theorem}
\label{thm:rateGirth}
If $G$ is a graph with girth $g$, then the file size $M(k)$ of an FR code $C_G$ based on $G$ satisfies
\[M(k)=\left\{\begin{array}{cc}
         k\alpha -k+1 \;&\textmd{ if } k\leq g-1 \\
           k \alpha-k \;&\textmd{ if }g\leq k\leq g+\lceil\frac{g}{2}\rceil-2.
         \end{array}\right.
\]
\end{theorem}
A $(d,g)$-\emph{cage} is a $d$-regular graph with girth $g$ and minimum number of vertices.
Let $N(d,g)$ be the minimum number of vertices in a $(d,g)$-cage.
A lower bound on $N(d,g)$, known as \emph{Moore bound}~\cite[p. 180]{Biggs}, is given by
\begin{equation*}
%\label{eq:cage}
   n_0(d,g)=\left\{\begin{array}{c c}
         1+d\sum_{i=0}^{\frac{g-3}{2}}(d-1)^i &\;\textmd{ if }g \textmd{ is odd} \\
            2\sum_{i=0}^{\frac{g-2}{2}}(d-1)^i&\;\textmd{ if }g \textmd{ is even}
         \end{array}\right..
       \end{equation*}
%Next, we use the Moore bound to show that the bound in~(\ref{eq:bound2}) can be improved in some cases.
\begin{lemma}
\label{lm:not tight}
The  bound in~(\ref{eq:bound2}) is not tight  for $\rho=2$ if $$\alpha k-\alpha-k+3\leq n< N(\alpha,k+1).$$
% where
%$ N_0(\alpha,g)$ is defined in Theorem~\ref{MooreBound}.
\end{lemma}
As a consequence of Lemma~\ref{lm:not tight} we have that the bound in~(\ref{eq:bound2}) is not always tight and hence we have a similar better bound on $A(n,k,\alpha, \rho)$:
 \begin{equation*}
   A(n,k,\alpha, \rho)\leq  \varphi'(k), \textmd{ where } \varphi'(1)=\alpha,
  \end{equation*}
  \[\textmd{  } \varphi'(k+1)= A(n,k,\alpha, \rho)+\alpha-\left\lceil\frac{\rho  A(n,k,\alpha, \rho)-k\alpha}{n-k}\right\rceil.\]
%*********************************************************************************************************************************************************
%*********************************************************************************************************************************************************
%                                                         rho>2
%*********************************************************************************************************************************************************
%*********************************************************************************************************************************************************

\section{Optimal FR Codes with Repetition Degree $\rho > 2$}
\label{sec:rho>2}

In this section, we consider FR codes with repetition degree $\rho> 2$. Note, that while  codes with $\rho=2$ have the maximum data/storage ratio,  codes with $\rho>2$ provide multiple choices for node repairs. In other words, when a node fails, it can be repaired from different $d$-subsets of available nodes.

We present  generalizations  of the  constructions from the previous section which were based on Tur\'an graphs and graphs with a given girth. These generalizations employ transversal designs and generalized polygons, respectively.

A \emph{transversal design} of group size $h$ and block size $\ell$,  denoted by $\text{TD}(\ell, h)$
%$\text{TD}_{\lambda}(\ell, h)$
is a triple $(\cP,\mathcal{G},\mathcal{B})$, where

\begin{enumerate}
\item $\cP$ is a set of $\ell h$ \emph{points};

\item $\mathcal{G}$ is a partition of $\cP$ into $\ell$ sets
(\emph{groups}), each one of size $h$;

\item $\mathcal{B}$ is a collection of $\ell$-subsets of $\cP$
(\emph{blocks});

\item each block meets each group in exactly one point;

%\item each $t$-subset of points that meets each group in at most
%one point is contained in exactly $\lambda$ blocks.
%\item any pair of points from different groups is contained in exactly $\lambda$ blocks.
\item any pair of points from different groups is contained in exactly one block.
\end{enumerate}
 The properties of a transversal design $\text{TD}(\ell,h)$ which will be useful for our constructions  are summarized in the following lemma~\cite{Anderson}.

\begin{lemma}
\label{lm:TDparameters}
Let $(\cP,\cG,\cB)$ be a transversal design $\text{TD}(\ell,h)$. The number of points is given by $|\cP|=\ell h$, the number of groups is given by $|\cG|=\ell$, the number of blocks is given by $|\cB|=h^2$, and the number of blocks that contain a given point is equal to $h$.
%\begin{itemize}
%  \item The number of points is given by $|\cP|=\ell h$;
%  %\item The number of groups is given by $|\cG|=\ell$;
%  \item The number of blocks is given by $|\cB|=h^2$;
%  \item The number of blocks that contain a given point is equal to $h$.
%  \item The girth of the incidence graph of a transversal design is equal to $6$.
%\end{itemize}
\end{lemma}
%A $\text{TD}(\ell,h)$ is called \emph{resolvable} if the set $\mathcal{B}$
%can be partitioned into subsets $\mathcal{B}_1,...,\mathcal{B}_h$, each one contains $h$ blocks,
%such that each element of $\cP$ is contained in exactly one block of
%each $\mathcal{B}_i$, i.e. the blocks of $\cB_i$ partition the set $\cP$.
%%The sets $\mathcal{B}_1,...,\mathcal{B}_h$ are called \emph{parallel classes}.
%Resolvable transversal design $\text{TD}(\ell,q)$ is known to exist for any $\ell \leq q$ and prime power $q$~\cite{Anderson}.

 Let $\textmd{TD}$ be a transversal design $\textmd{TD}(\rho,\alpha)$, $\rho\leq \alpha+1$, with block size $\rho$ and group size $\alpha$. Let $C_{\textmd{TD}}$ be an $(n,\alpha,\rho)$ FR code based on $\textmd{TD}$ (see Section~\ref{sec:DRESS}).
 By Lemma~\ref{lm:TDparameters}, there are $\rho \alpha$ points in $\textmd{TD}$ and hence  $n= \rho \alpha$. Note, that  all the symbols stored in  node $i$ correspond to the set $N_i$ of blocks from $\textmd{TD}$ that contain the point $i$. Since by Lemma~\ref{lm:TDparameters} there are $\alpha$  blocks that contain a given point, it follows that each node stores $\alpha$ symbols.

%The proof of the following theorem can be found in an extended version of this paper~\cite{SiEtFR}.
% Similarly to  Theorem~\ref{thm:alternativeTuran} we can prove the following theorem.
%
%One can prove that the rate $R_{C_{\textmd{TD}}}$ is similar to the rate $R_{C_T}$  of an $(\rho\alpha,\rho)$-Tur\'an graph, with different $\alpha$'s for these two FR codes (see Theorem~\ref{trm:Turan}).

\begin{theorem}
\label{lm:TDrate}
Let $k=b\rho+t$, for integers $b,t\geq 0$ such that $t\leq \rho-1$. For an $(n=\rho\alpha,\alpha,\rho)$ FR code $C_{\textmd{TD}}$  based on a transversal design $\textmd{TD}(\rho,\alpha)$ we have
$$M(k)\geq k\alpha-\binom{k}{2}+\rho\binom{b}{2}
+bt.
$$
\end{theorem}
%\begin{proof}Similar to the proof of Theorem~\ref{trm:Turan}.
%\end{proof}
\begin{remark} Note, that for all $k\geq \rho+1$, the file size of the FR code $C_{\textmd{TD}}$ is strictly larger than the MBR capacity.
\end{remark}

Note that the incidence matrix of the transversal design $\textmd{TD}(2,\alpha)$ is equal to the incidence matrix of the $(2\alpha, 2)$-Tur\'an graph, and hence in this case $C_{TD}=C_{T}$.
%Since the incidence matrix of the transversal design $\textmd{TD}(2,\alpha)$ is equal to the incidence matrix of the $(2\alpha, 2)$-Tur\'an graph, one can prove the following corollary.
%\begin{corollary}
%Let $C_{\textmd{TD}}$ be an $(r\alpha,\alpha,r)$ FR code based on $\textmd{TD}$, a transversal design $\textmd{TD}(r,\alpha)$. Let $C_T$ be an $(\frac{r}{r-1}\alpha,\alpha,2)$ FR code based on $(n,r)$-Tur\'an graph $T$. If $M_{\textmd{TD}}(k)$ and $M_{T}(k)$ are their file sizes, respectively, then
%\begin{enumerate}
%  \item $C_{\textmd{TD}}=C_T$ for $r=2$;
%  \item $M_{\textmd{TD}}(k)\geq M_{T}(k)$ for all $r\geq 2$.
%\end{enumerate}
%\end{corollary}
\begin{example}
\label{ex:TD_34}
Let TD be a transversal design $\textmd{TD}(3,4)$  defined as follows:
 $\cP=\{1,2,\ldots,12\}$; $\mathcal{G}=\{G_1,G_2,G_3\}$, where
$G_1=\{1,2,3,4\}$, $G_2=\{5,6,7,8\}$, and $G_3=\{9,10,11,12\}$;
$\mathcal{B}=\{B_1,B_2,\ldots, B_{16}\}$,
 where
$B_1=\{1,5,9\}$, $B_2=\{1,6,10\}$, $B_3=\{1,7,11\}$, $B_4=\{1,8,12\}$,
$B_5=\{2,5,10\}$, $B_6=\{2,6,9\}$, $B_7=\{2,7,12\}$, $B_8=\{2,8,11\}$,
$B_9=\{3,5,12\}$, $B_{10}=\{3,6,11\}$, $B_{11}=\{3,7,10\}$, $B_{12}=\{3,8,9\}$,
$B_{13}=\{4,5,11\}$, $B_{14}=\{4,6,12\}$, $B_{15}=\{4,7,9\}$, and $B_{16}=\{4,8,10\}$.
%with incidence matrix  given by

%\begin{footnotesize}
%\[\textbf{I}(\textmd{TD})=\left(
%                   \begin{array}{cccccccccccccccc}
%                     1 & 1 & 1 & 1 & &&&&&&&&&&& \\
%                     &&& & 1 & 1 & 1 & 1 & &&&&&&& \\
%                     &&&&&&&& 1 & 1 & 1 & 1 & &&& \\
%                     &&&&&&&&&&& & 1 & 1 & 1 & 1 \\
%                     1 & && & 1 & && & 1 & && & 1 & && \\
%                      & 1 & &&& 1 & && & 1 & && & 1 & & \\
%                      &  & 1 & && & 1 & && & 1 & && & 1 &  \\
%                      & & & 1 & && & 1 & && & 1 & && & 1 \\
%                     1 & &&& & 1 & &&&& & 1 & & & 1 &  \\
%                      & 1 & & & 1 & &&&& & 1 & &&& & 1 \\
%                      &  & 1 & &&& & 1 &  & 1 & & & 1 & && \\
%                      & & & 1 & & & 1 &  & 1 & &&& & 1 & & \\
%                   \end{array}
%                 \right)
%\]
%\end{footnotesize}

The placement of symbols from a codeword of the corresponding MDS code of length $16$ is shown in Fig.~\ref{fig:TD34}.
\begin{figure}[t]
 \centering
 \includegraphics[width=0.80\columnwidth]{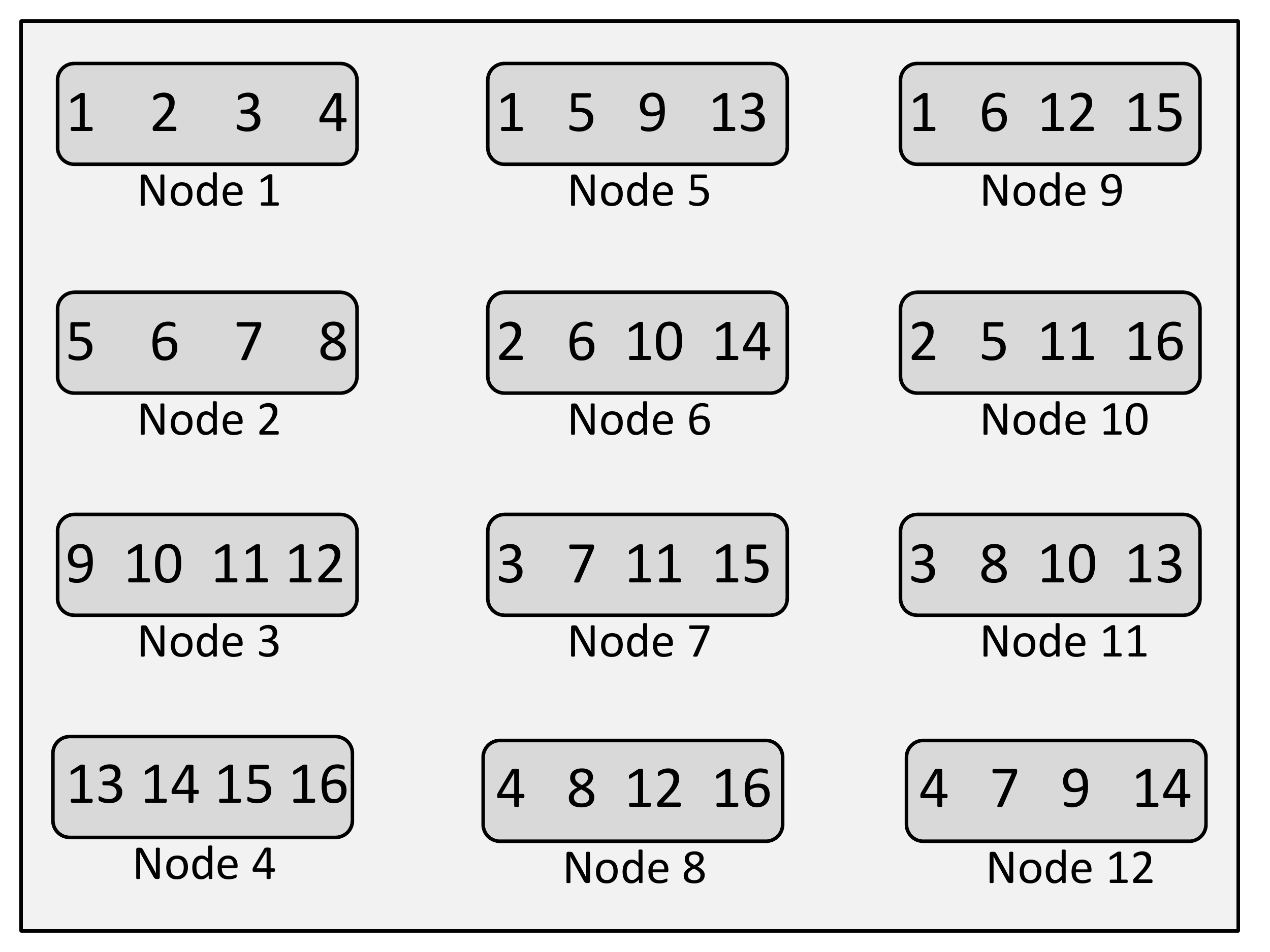}
 \caption{The $(12,4,3)$ FR code based on \textmd{TD}(3,4)}\label{fig:TD34}
\end{figure}
The  values of a file size $M(k)$ for $1\leq k\leq 4$ are given in the following~table.
\[\begin{tabular}{|c|c|}
    \hline
    $k$ & $M(k)$ \\  \hline\hline
    1 & 4 \\\hline
    2 & 7 \\\hline
    3 & 9 \\\hline
    4 & 11 \\
    \hline
  \end{tabular}
\]

\end{example}

\begin{remark}
The conditions on the parameters of TD such that the bound on the file size of an FR code $C_{\textmd{TD}}$ from Theorem~\ref{lm:TDrate} attains the recursive bound in~(\ref{eq:bound2}) can be found in~\cite{SiEtFR}.
\end{remark}

Similarly to an FR code $C_{G}$ with $\rho=2$ based on a graph~$G$ with girth $g$, one can consider an FR code $C_{\textmd{GP}}$ based on a generalized $g$-gon (generalized polygon GP~\cite{Biggs}) for $\rho>2$. One can prove that the file size of $C_{\textmd{GP}}$ is identical to the file size of $C_{G}$ for $k\leq g+\lceil\frac{g}{2}\rceil-2$ given in Theorem~\ref{thm:rateGirth}. However, a generalized $g$-gon  is known to exist only for $g\in\{3,4,6,8\}$.
This observation also holds for a general biregular bipartite graph of girth $2g$, not only the incidence graph of a generalized polygon.

\begin{remark}
\label{rm:expander-FR}
 Note that the problem of constructing FR codes with $\rho>2$ also can be considered in terms of bipartite \emph{expander} graphs (see e.g~\cite{GUV07}). Let $G_{Ex}=(L\cup R,E)$ be a bipartite  expander and let $C_{Ex}$ be the FR code such that the subset $N_i$, $1\leq i\leq n$, corresponds to the $i$th vertex in $L$ and the symbol $j$, $1\leq j\leq \theta$, corresponds to the $j$th vertex in $R$, $|L|=n$ and $|R|=\theta$.
 Then calculating $M(k)$ can be described by calculating the number of neighbours of any subset of $L$ of size $k$. In other words, for an FR code with  file size $M(k)$ it should hold that $|\Gamma(A)|\geq M(k)$ for every $A\subseteq L$ of size $k$, where $\Gamma(A)$ denotes the set of neighbours of~$A$. Hence, to have an FR code with file size $M(k)$, one need to construct a $(k,\frac{M(k)}{k})$ expander graph, where $\frac{M(k)}{k}$ is its expansion factor~\cite{GUV07}.
\end{remark}
 %The existence of such graphs was considered in~\cite{AABL12,AABLL13,ABV09,ExJa08,FLSUW95}.

%*********************************************************************************************************************************************************
%*********************************************************************************************************************************************************
%                                                         FRB codes
%*********************************************************************************************************************************************************
%*********************************************************************************************************************************************************

\section{Fractional Repetition Batch Codes}
\label{sec:FRB}

In this section we propose a new type of codes for DSS, called \emph{fractional repetition batch}  (FRB) codes, which enable uncoded efficient exact node repairs and load balancing which can be performed by several users in parallel. An FRB code is a combination of an FR code and an uniform combinatorial batch code.

The family of codes called \emph{batch codes}  was proposed  in~\cite{IKOS04} for load balancing in distributed storage.
A \emph{batch code} stores $\theta$ (encoded) data symbols in $n$ system nodes in such a way that any batch of $t$ data symbols can be decoded by reading at most one symbol from each node.
In a $\rho$-\emph{uniform} \emph{combinatorial batch code}, proposed in~\cite{PSW08},
each node stores a subset of data symbols and no decoding is required during retrieval of any batch of $t$ symbols. Each symbol is stored in exactly $\rho$ nodes and hence it is also called \emph{a replication based} batch code. A $\rho$-uniform combinatorial batch code is denoted by $\rho-(\theta,N,t,n)$-CBC, where $N=\rho \theta$ is the  total storage over all the $n$ nodes. These codes were studied in~\cite{IKOS04,PSW08,SiGa13}.

Next, we provide a formal definition of FRB codes. This definition is based on the definitions of a DRESS code and a uniform combinatorial batch code.
Let $\bf f$ $\in \F_q^M$ be a file of size $M$ and let $c_{\bf f}\in \F_q^{\theta}$ be a codeword of an $(\theta,M)$ MDS code which encodes the data $\bf f$.
Let $\{N_1,\ldots,N_n\}$ be a collection of $\alpha$-subsets of the set $[\theta]$.
 A $\rho-(n,M,k,\alpha,t)$ \emph{FRB code} $C$, $k\leq \alpha$, $t\leq M$, represents a system of $n$ nodes with the following properties:
\begin{enumerate}
  \item Every node $i$, $1\leq i\leq n$, stores $\alpha$ symbols of $c_{\bf f}$ indexed by $N_i$;
  \item Every symbol of $c_{\bf f}$ is stored in $\rho$ nodes;
  \item From any set of $k$ nodes it is possible to reconstruct the stored file $\bf f$, in other words, $M=\min_{|I|=k}|\cup_{i\in I}N_i|$;
  \item Any batch of $t$ symbols from $c_{\bf f}$ can be retrieved by downloading at most one symbol from each node.
\end{enumerate}
%Note that the total storage over all $n$ nodes needed to store a file $\bf f$ equals to $n\alpha=\theta\rho$.
Note that the retrieval of any batch of $t$ symbols can be performed by $t$ different users in parallel, where each user gets a different symbol.

%\begin{remark}
%Note that while in a classical batch code any $t$ \emph{data} symbols can be retrieved, in an FRB code any batch of $t$ \emph{coded} symbols can be retrieved.  In particular, when a systematic MDS code is chosen for an FRB code, the data symbols can be easily retrieved.
%\end{remark}

In the following, we present our constructions of FRB codes which are based on the uniform batch codes from~\cite{PSW08} and~\cite{SiGa13} and on FR codes considered in Sections~\ref{sec:rho2} and ~\ref{sec:rho>2}.

\begin{theorem}
\label{thm:FRB from TD}
$~$
\begin{enumerate}
  \item If $K_{\alpha,\alpha}$ is a complete bipartite graph with $\alpha>2$, then $C_{K_{\alpha,\alpha}}$ is a $2-(2\alpha, M,k,\alpha,5)$ FRB code with
  $M=k\alpha-\left\lfloor\frac{k^2}{4}\right\rfloor$.
  \item If $G$ is an $\alpha$-regular graph on $n$ vertices with girth~$g$, then $C_G$ is a $2-(n, M, k,\alpha,2g-\lfloor \frac{g}{2}\rfloor-1)$ FRB code with
  \[M=\left\{\begin{array}{cc}
         k\alpha -k+1 \;&\textmd{ if } k\leq g-1 \\
           k \alpha-k \;&\textmd{ if }g\leq k\leq g+\lceil\frac{g}{2}\rceil-2.
         \end{array}\right.
\]
  \item Let TD be a resolvable transversal design $\textmd{TD}(\alpha-1,\alpha)$,  for a prime power $\alpha$.  $C_{\textmd{TD}}$ is an $(\alpha-1)-(\alpha^2-\alpha, M,k,\alpha,\alpha^2-\alpha-1)$ FRB code with
  $M\geq k\alpha -\binom{k}{2}+(\alpha-1)\binom{x}{2}+xy$, where $x,y$ are nonnegative integers which satisfy $k=x(\alpha-1)+y$, $y\leq \alpha-2$.
\end{enumerate}
\end{theorem}

\begin{example}
{~}
\begin{itemize}
\item Consider the code $C_{K_{3,3}}$ based on $K_{3,3}$ (see also Example~\ref{ex:bipartite} for an FR code based on $K_{3,3}$).
By Theorem~\ref{thm:FRB from TD}, for $k=3$, $C_{K_{3,3}}$ is a $2-(6, 7,3,3,5)$ FRB code.

\item Consider the code $C_{\textmd{TD}}$  based on the resolvable transversal design $\textmd{TD}=\textmd{TD}(3,4)$ (see also Example~\ref{ex:TD_34} for an FR code based on $\textmd{TD}(3,4)$). By Theorem~\ref{thm:FRB from TD}, for $k=4$, $C_{\textmd{TD}}$ is a $3-(12, 11,4,4,11)$ FRB code, which stores a file of size $11$ and allows for retrieval of  any (coded) $11$ symbols, by reading at most one symbol from a node. In particular, when using a systematic MDS code, $C_{\textmd{TD}}$  provides load balancing in data reconstruction.
\end{itemize}

\end{example}

%\begin{remark}
% The structure of the incidence matrix of an FRB code and some other constructions are considered in~\cite{Sil14}.
%\end{remark}
%\hfill

\begin{remark} Similarly to FR codes, the problem of constructing for FRB codes can be considered in terms of bipartite expanders (see Remark~\ref{rm:expander-FR}).
The construction of batch codes based on (unbalanced) expander graphs was proposed in~\cite{IKOS04}. To construct an FRB code, one need a bipartite expander with two different expansion factors, $M(k)/k$ and 1, for two sides $L$ and $R$ of a graph, respectively.
\end{remark}

%\subsection{Subsection Heading Here}
%Subsection text here.
%
%
%\subsubsection{Subsubsection Heading Here}
%Subsubsection text here.
%
%\vspace{-.2cm}

\section{Conclusion}
\label{sec:conclusion}
We considered the problem of constructing  optimal  FR codes and as a consequence, optimal DRESS codes.
We presented constructions of FR codes based on Tur\'an graphs,
graphs with a given girth, transversal designs, and generalized polygons.
%regular graphs, such as Tur\'an graphs and
%graphs with a given girth; and constructions of FR codes based  on combinatorial designs, such as transversal designs and generalized polygons.
Based on a connection between FR codes and batch codes, we proposed a new family of codes for DSS, FRB codes, which have the properties of batch codes and FR codes simultaneously.  These are the first codes for DSS which allow  uncoded efficient exact repairs and load balancing.
%We %presented  examples of constructions for FRB codes, based on transversal designs, complete bipartite graphs, and graphs with large girth.

%
%
%
%
%
\section*{Acknowledgment}
This research was supported in part by the Israeli Science Foundation (ISF), Jerusalem, Israel, under Grant no.~10/12.
Natalia Silberstein was also supported in part at the Technion by a Fine Fellowship.


%\vspace{-.2cm}

%

\end{document}